\documentclass[12pt,reqno]{amsart}
\usepackage{amscd,amssymb,amsmath,amsthm}
\usepackage[arrow,matrix]{xy}

\usepackage{graphicx}
\usepackage{epsfig}
\usepackage{epstopdf}
\usepackage{epsfig}
\usepackage{amsmath}
\usepackage{amsthm}
\usepackage{blindtext}

\def\cb{{\mathcal B}}

\def\bb{{\mathbb B}}

\def\a{\alpha}
\def\b{\beta}

\def\l{\lambda}

\def\s{\sigma} 

\def\v{\varphi} \def\F{\Phi}

 \def\Om{\Omega}
\def\ab{\overline{a}} \def\bb{\overline{b}} \def\cb{\overline{c}}

\def\xd{x_{\downarrow}}
\def\Cbd{C_{b_{\downarrow}}}
\def\Cb{C_b}

\def\y{\gamma}

\newtheorem{thm}{Theorem}[section]
\newtheorem{lem}[thm]{Lemma}

\newtheorem{defin}[thm]{Definition}

\newtheorem{rem}{Remark}[section]

\begin{document}
\title[Periodic and Weakly Periodic Ground States]{Periodic and Weakly Periodic Ground States for the $\lambda$-Model on Cayley Tree}

\author{Farrukh Mukhamedov}
\address{Farrukh Mukhamedov\\
 Department of Mathematical Sciences\\
College of Science, The United Arab Emirates University\\
P.O. Box, 15551, Al Ain\\
Abu Dhabi, UAE} \email{{\tt far75m@gmail.com} {\tt
farrukh.m@uaeu.ac.ae}}

\author{Chin Hee Pah}
\address{Chin Hee Pah\\
Department of Computational \& Theoretical Sciences,\\
Faculty of Science, International Islamic University Malaysia,\\
P.O. Box, 141, 25710, Kuantan, Pahang, Malaysia} \email{{\tt
pahchinhee@iium.edu.my}}

\author{Muzaffar Rahmatullaev}
\address{Muzaffar Rahmatullaev\\
Institute of Mathematics, National University of Uzbekistan, 29,
Do'rmon Yo'li str., 100125, Tashkent, Uzbekistan.} \email {{\tt
mrahmatullaev@rambler.ru}}

\author{Hakim Jamil}
\address{Hakim Jamil\\
 Department of Computational \& Theoretical Sciences,\\
Faculty of Science, International Islamic University Malaysia,\\
P.O. Box, 141, 25710, Kuantan, Pahang, Malaysia} \email{{\tt
m.hakimjamil@yahoo.com.my}}

\begin{abstract}
In this paper we consider the $\lambda$-model on the Cayley tree of
order two. We describe periodic and weakly periodic ground states
for the considered model.
\end{abstract}

\maketitle

\section{Introduction}

Each Gibbs measure is associated with a single phase of a physical
system \cite{11}. Existence of two or more Gibbs measures means that
phase transitions exist. One of fundamental problems is to describe
the extreme Gibbs measures corresponding to a given Hamiltonian.

As is known, the phase diagram of Gibbs measures for a Hamiltonian
is close to the phase diagram of isolated (stable) ground states of
this Hamiltonian. At low temperatures, a periodic ground state
corresponds to a periodic Gibbs measure, see \cite{5}, \cite{11}.
The problem naturally arises on description of periodic and weakly
periodic ground states. For the Potts model with competing
interactions on the Cayley tree of order two, periodic ground states
were studied in \cite{GMM06,1,MRM07} (see also \cite{Roz}). The
notion of a weakly periodic ground state was introduced in
\cite{10}. For the Ising model with competing interactions, weakly
periodic ground states were described in \cite{6}, \cite{10}. For
the Potts model, such states for normal subgroups of index 2 were
studied in \cite{7}. For the Potts model, such states for normal
subgroups of index 4 were studied in \cite{12}.

On the other hand, it is natural to consider a model which is more
complicated than Potts one, therefore, in \cite{M} it was proposed
to study so-called $\l$-model on the Cayley tree (see also
\cite{13,Roz}). In that model, many possible interactions
(nearest-neighbor) are taken into account. In the mentioned paper,
for a special kind of $\l$-model, its disordered phase has been
studied (see \cite{MR3}) and some its algebraic properties were
investigated. Furthermore, in \cite{Hak} we have described ground
states of the $\l$-model on the Cayley tree of order two. This model
is much more general than Potts model, and exhibits interesting
structure of ground states.

In the present paper, we continue an investigation of the
$\lambda$-model on the Cayley tree of order two and study periodic
and weakly periodic ground states for normal subgroups of index 2.
The paper is organized as follows. In Section 2, we recall the main
definitions and known facts. In Section 3, we study periodic and
weakly periodic ground states. Note that our results extend known
ones \cite{RKh2013,10}. Hence, the proposed model opens new
perspectives in the theory of statistical mechanics over trees.

\section{Preliminaries}
Let $\tau^k = (V,L)$ be a Cayley tree of order
$k$, i.e, an infinite tree such that exactly $k+1$ edges are incident to each vertex. Here $V$ is
the set of vertices and $L$ is the set of edges of $\tau^k$.

Let $G_k$ denote the free product of $k+1$ cyclic groups $\{e, a_i\}$ of order 2 with generators $a_1, a_2,\dots, a_{k+1}$,
i.e., let $a^2_i=e$ (see\cite{karga}).

There exists a one-to-one correspondence between the set $V$ of vertices of the Cayley tree of order $k$ and the group $G_k$, see\cite{Gan},\cite{Roz}.

We show how to construct this correspondence. We choose an arbitrary vertex $x_0 \in V  $and associate it with the identity element $e$ of the group $G_k$. Since we may assume that the graph under consideration is planar, we associate each neighbor of $x_0$ (i.e., $e$) with a single generator $a_i, i=1, 2,\dots, k + 1$, where
the order corresponds to the positive direction, see Figure \ref{cayley}.
\begin{figure}[h!]
    \begin{center}
        \includegraphics[width=12cm]{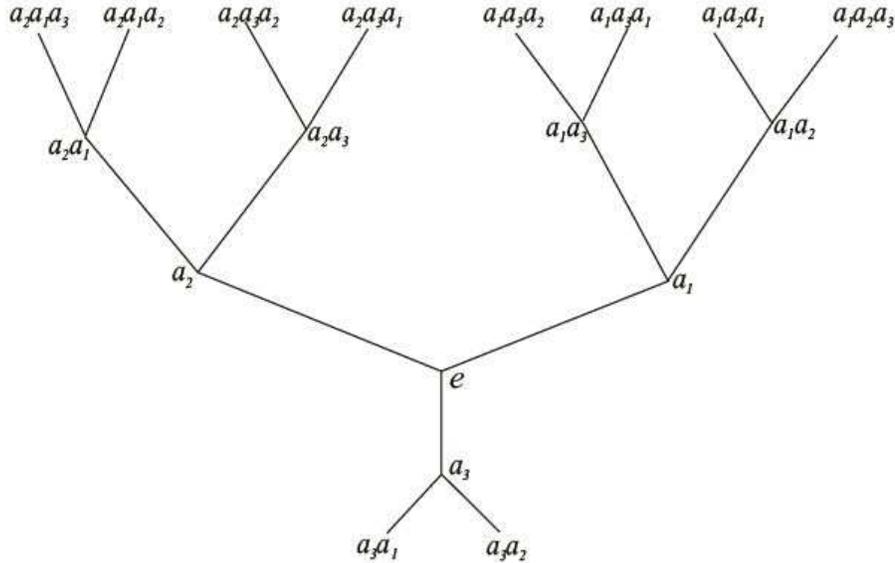}
    \end{center}
    \caption{The Cayley tree $\tau^2$ and elements of the group representation of vertices} \label{cayley}
\end{figure}

For every neighbor of $a_i$, we introduce words of the form $a_ia_j$ . Since one of the neighbors of $a_i$ is $e$, we put $a_ia_i = e$. The remaining neighbors of $a_i$ are labeled according to the above order. For every neighbor of $a_ia_j$ , we introduce words of length 3 in a similar way. Since one of the neighbors of $a_ia_j$ is $a_i$, we put $a_ia_ja_j = a_i$. The remaining neighbors of $a_ia_j$ are labeled by words of the form $a_ia_ja_l$, where $i, j, l = 1, 2, . . . , k + 1$, according to the above procedure. This agrees with the previous stage
because $a_ia_ja_j = a_ia^2_j= a_i$. Continuing this process, we obtain a one-to-one correspondence between the vertex set of the Cayley tree $\tau^k$ and the group $G_k$.

The representation constructed above is said to be $right$ because,
for all adjacent vertices $x$ and $y$ and the corresponding elements
$g,h \in G_k,$ we have either $g = ha_i$ or $h=ga_j$ for suitable
$i$ and $j$. The definition of the $left$ representation is similar.

For the group $G_k$ (or the corresponding Cayley tree), we consider the left (right) shifts. For $g \in G_k$,
we put

\begin{equation*}
    T_g(h)=gh \ (T_g(h)=hg) \ for \ all \ h \in G^k
\end{equation*}
The group of all left (right) shifts on $G_k$ is isomorphic to the group $G_k$.

Each transformation $S$ on the group $G_k$ induces a transformation $S$ on the vertex set $V$ of the Cayley
tree  $\tau^k$. In the sequel, we identify $V$ with $G_k$.

The following assertion is quite obvious (see [2, 9]).

\begin{thm}
    The group of left (right) shifts on the right (left) representation of the Cayley tree is the group of translations.
\end{thm}

By the group of translations we mean the automorphism group of the
Cayley tree regarded as a graph. Recall (see, for example,
\cite{Melni}) that a mapping $\psi$ on the vertex set of a graph G
is called an automorphism of G if $\psi$ preserves the adjacency
relation, i.e., the images $\psi(u)$  and $\psi(v)$ of vertices $u$
and $v$ are adjacent if and only if $u$ and $v$ are adjacent.

For an arbitrary vertex $x_0 \in V $, we put

\begin{equation*}
    W_{n}=\left\{ x\in V\mid d(x^0,x)=n\right\}, \ \
    V_n=\bigcup\limits_{m=0}^{n}W_m , \ \
    L_{n}=\left\{
    l=<x,y>\in L\mid x,y\in V_{n}\right\}.
\end{equation*}
where $d(x, y)$ is the distance between $x$ and $y$ in the Cayley tree, i.e., the number of edges of the path
between $x$ and $y$.

For each $x \in G_k$, let $S(x)$ denote the set of immediate successor of $x$, i.e., if $x \in W_n$ then

$$
S(x)=\left\{ y\in W_{n+1}:d(x,y)=1\right\}.
$$

For each $x \in G_k$, let $S_1(x)$ denote the set of all neighbors of $x$, i.e., $S_1(x)=\{y \in G_k:<x,y> \in \ L \}$. The set $S_1(x) \setminus S(x)$ is a singleton. Let $\xd$ denote the (unique) element of this set.

Assume that spin takes its values in the set $\F=\{1, 2,\dots, q\}.$
By a configuration $\s$ on $V$ we mean a function taking $\sigma:x
\in V\to \s(x) \in \F.$ The set of all configurations coincides with
the set $\Omega = \F^V$.

Consider the quotient group $G_k/G^*_ k = \{H1,\dots,H_r\}$, where
$G^*_k$ is a normal subgroup of index $r$ with $r \geq 1$.

\begin{defin}
A configuration $\s(x)$ is said to be $G^*_k$-periodic if $\s(x)
=\s_i$ for all $x \in G_k$ with $x \in H_i$. A $G_k$-periodic
configuration is said to be translation invariant.
\end{defin}
By \textit{period} of a periodic configuration we mean the index of
the corresponding normal subgroup.

\begin{defin}
    A configuration $\s(x)$ is said to be $G^*_k$-weakly periodic if $\s(x) =\s_{ij}$ for all $x \in G_k$ with $\xd \in H_i$ and $x \in H_j$.
\end{defin}

The Hamiltonian of the $\l$-model \cite{M,Hak} has a form,
\begin{eqnarray}\label{h}
H(\sigma)=\sum\limits_{<x,y>\in L}\l(\sigma (x),\sigma (y))
\end{eqnarray}

In what follows, we restrict ourself to the case $k=2$ and $\F=\{1,2,3\}$, and for the sake of simplicity, we consider the following function:
\begin{equation}\label{cond}
\l(i,j)=\left\{
\begin{array}{lll}
\ab&if&|i-j|=2,\\
\bb&if&|i-j|=1,\\
\cb&if&i=j,
\end{array}\right.
\end{equation}
where $\ab,\bb,\cb\in \mathbb{R}$ for some given numbers.

\section{Ground States}

In this section, we describe ground state of the $\l$-model on a Cayley tree. For a pair of configurations $\s$ and $\v$ coincinding almost everywhere, i.e., everywhere except finitely many points,we consider the relative Hamiltonian $H(\s,\v)$ determining the energy differences of the configurations $\s$ and $\v$:

\begin{eqnarray}\label{eq12}
H(\s,\v)=\sum_{\substack{<x,y>\\ x,y \in V}}(\l(\s(x),\s(y))-\l(\v(x),\v(y)))
\end{eqnarray}

Let $M$ be the set of unit balls with vertices in $V$, i.e. $M=\{x, S_1(x), \forall x\in V$\}. We call
the restriction of a configuration $\s$ to the ball $b\in M$ a
{\it bounded configuration} $\s_b .$

We define the energy of the configuration $\s_b$ on $b$ as follows
\[U(\s_b)=\dfrac{1}{2}\sum_{\substack{<x,y>\\ x,y \in V}}(\l(\s(x),\s(y)))\]
From \eqref{eq12}, we got the following lemma.

We shall say that two bounded configurations $\s_b$ and $\s'_{b'}$
belong to the same class if $U(\s_b)=U(\s'_{b'})$ and we write
 $\s'_{b'}\sim \s_b. $

\begin{lem}
    Relative Hamiltonian \eqref{eq12} has the form

    \[H(\s,\v)=\sum_{b \in M}(U(\s_b)-U(\v_b)).\]
\end{lem}

\begin{lem}
    The inclusion
    \begin{eqnarray}\label{eq13}
    U(\v_b)\in \left\{\dfrac{\a + \b + \y}{2},\forall \ \a,\b,\y \in \{\ab,\bb,\cb\}\right\}
    \end{eqnarray}
    holds for every configuration $\v_b$ on b$(b \in M)$.
\end{lem}

\begin{defin}
    A configuration $\v$ is called a ground state of the relative Hamiltonian H if
    \begin{eqnarray}\label{eq14}
    U(\v_b)= \min \left\{\dfrac{\a + \b + \y}{2},\forall \ \a,\b.\y \in \{\ab,\bb,\cb\}\right\}
    \end{eqnarray}
    for any $b \in M$

\end{defin}

If a ground state is a periodic (resp. weakly periodic, translation
invariant) configuration then we call it a periodic (resp. weakly
periodic, translation invariant) ground state.

For any configuration $\s_b$, we have
\[U(\s_b)\in\{U_1,U_2,U_3,U_4,U_5,U_6,U_7,U_8,U_9,U_{10}\},\]

where
\begin{eqnarray}\label{U_n}
U_1=3\ab/2,&U_2=3\bb/2,\nonumber\\
U_3=3\cb/2,&U_4=\ab/2+\bb,\nonumber\\
U_5=\ab/2+\cb,&U_6=\bb/2+\ab,\nonumber\\
U_7=\bb/2+\cb,&U_8=\cb/2+\ab,\nonumber\\
U_9=\cb/2+\bb,&U_{10}=(\ab+\bb+\cb)/2.
\end{eqnarray}
We denote
\begin{eqnarray}\label{eq15}
A_m=\{(\ab,\bb,\cb)\in \mathbb{R}^3|\  U_m=\min_{1 \leq k \leq 10}\{U_k\}\}
\end{eqnarray}
Using \eqref{eq15}, we obtain
$$\begin{array}{llllllllll}
A_1=\left\{(\ab,\bb,\cb)\in \mathbb{R}^3| \ \ab \leq \bb\leq \cb\right\}\cup \left\{(\ab,\bb,\cb)\in \mathbb{R}^3| \ \ab \leq \cb \leq \bb\right\},\\
A_2=\left\{(\ab,\bb,\cb)\in \mathbb{R}^3| \ \bb \leq \cb\leq \ab\right\}\cup \left\{(\ab,\bb,\cb)\in \mathbb{R}^3| \ \bb \leq \ab\leq \cb\right\},\\
A_3=\left\{(\ab,\bb,\cb)\in \mathbb{R}^3| \ \cb \leq \bb\leq \ab\right\}\cup \left\{(\ab,\bb,\cb)\in \mathbb{R}^3| \ \cb \leq \ab \leq \bb\right\},\\
A_4=\left\{(\ab,\bb,\cb)\in \mathbb{R}^3| \ \ab=\bb \leq \cb\right\},\\
A_5=\left\{(\ab,\bb,\cb)\in \mathbb{R}^3| \ \ab=\cb \leq \bb\right\},\\
A_6=\left\{(\ab,\bb,\cb)\in \mathbb{R}^3| \ \ab=\bb \leq \cb\right\},\\
A_7=\left\{(\ab,\bb,\cb)\in \mathbb{R}^3| \ \bb=\cb \leq \ab\right\},\\
A_8=\left\{(\ab,\bb,\cb)\in \mathbb{R}^3| \ \ab=\cb \leq \bb\right\},\\
A_9=\left\{(\ab,\bb,\cb)\in \mathbb{R}^3| \ \bb=\cb \leq \ab\right\},\\
A_{10}=\left\{(\ab,\bb,\cb)\in \mathbb{R}^3| \ \ab=\bb=\cb\right\}, \mathbb{R}^3=\cup_{i=1}^{10}A_i.\\
\end{array}$$

Let us introduce some notations. We put
$$C_i=\{\s_b \in \Om_b: \
U(\s_b)=u_i \}, \ \ \ i=\overline{1,10}$$ and $B^{(i)}=|\{x \in
S_1(k):\ \v_b(x)=i\}|$ for $i=\overline{1,3}$.

Let $A\subset \{1,2,...,k+1\}$, $H_A=\{x\in G_k: \sum_{j\in
A}w_j(x)-$even$\},$ where $w_j(x)$-is the number of letters $a_j$
in the word $x.$

It is obvious, that $H_A$ is a normal subgroup of index two
\cite{Roz}. Let $G_k/H_A=\{H_A,G_k\setminus H_A\}$ be the quotient
group. We set $H_0=H_A, H_1=G_k\setminus H_A$.

\section{Periodic Ground States}

In this section, we describe $H_0$-periodic ground states. We note
that each $H_0$-periodic configurations has the following form:
\begin{equation}\label{per} \s(x)=\left\{
\begin{array}{lll}
\s_1,&\textrm{if}&\ x\in H_0,\\
\s_2,&\textrm{if}&\ x\in H_1
\end{array}\right.
\end{equation}
where $\s_i \in \F=\{1,2,3\}$ ($i=1, 2$).

\begin{thm}\label{th1}
    Let k=2 and $|A|=1$. Then for the $\l$-model given by \eqref{cond} the following statements hold:
\begin{itemize}
\item[(i)] let $|\s_1-\s_2|=0$, then $H_0$-periodic ground state exists if and only if the parameters $\ab,\bb,\cb$ belong to  $A_3$;

\item[(ii)] If $|\s_1-\s_2|=1$, then  $H_0$-periodic ground state exists if and only if the parameters $\ab,\bb,\cb$ belong to  $A_7$;

\item[(iii)] If $|\s_1-\s_2|=2$, then $H_0$-periodic ground state exists if and only if the parameters $\ab,\bb,\cb$ belong to
$A_5$.
\end{itemize}
\end{thm}

\begin{proof}
Let us consider the following configuration
    \begin{equation}\label{trin1}
    \v(x)=\left\{
    \begin{array}{lll}
    i,&\textrm{if}&\ x\in H_0\\
    i,&\textrm{if}&\ x\in H_1,
    \end{array}\right.
    \end{equation} where $i=1,2,3$.

By $C_b$ we denote the center of $b\in M$. Let $C_b \in H_0$, then
we have
    \begin{equation*}
        \v_{b}(C_b)=i,\ B^{(i)}=3;
    \end{equation*}
Hence, $\v_{b}(x) \in C_3$, i.e. if $(\ab,\bb,\cb)\in A_3$ then the
corresponding configuration is a ground state.

Now consider the following configuration
    \begin{equation}\label{per1}
    \v(x)=\left\{
    \begin{array}{lll}
    i,&\textrm{if}&\ x\in H_0,\\
    j,&\textrm{if}&\ x\in H_1,
    \end{array}\right.
    \end{equation} where $|i-j|=1$

Assume that $C_b \in H_0$, the one finds
    \begin{equation*}
        \v_{b}(C_b)=i,\ B^{(i)}=2,\ B^{(j)}=1;
    \end{equation*}
    Hence, $\v_{b}(x) \in C_7$.

    Let $C_b \in H_1$, one has
    \begin{equation*}
    \v_{b}(C_b)=j,\ B^{(i)}=1,\ B^{(j)}=2;
    \end{equation*}
    Hence, $\v_{b}(x) \in C_7$.

We conclude that, if $(\ab,\bb,\cb)\in A_7$ then the corresponding
periodic configuration $\v(x)$ is a $H_0$-periodic ground state.

Let us consider the following configuration
    \begin{equation}\label{per1}
    \v(x)=\left\{
    \begin{array}{lll}
    i,&\textrm{if}&\ x\in H_0,\\
    j,&\textrm{if}&\ x\in H_1,
    \end{array}\right.
    \end{equation} where $|i-j|=2$

Let $C_b \in H_0$, we have
    \begin{equation*}
        \v_{b}(C_b)=i,\ B^{(i)}=2,\ B^{(j)}=1;
    \end{equation*}
    Hence, $\v_{b}(x) \in C_5$.

    Let $C_b \in H_1$, one gets
    \begin{equation*}
    \v_{b}(C_b)=j,\ B^{(i)}=1,\ B^{(j)}=2;
    \end{equation*}
So, $\v_{b}(x) \in C_5$. Therefore, if $(\ab,\bb,\cb)\in A_5$ then
the corresponding periodic configuration $\v(x)$ is $H_0$-periodic
ground state. This completes the proof.
\end{proof}

\begin{rem}
Notice that periodic configurations (ground states) described in (i)
Theorem \ref{th1} are translation-invariant. Full description of
translation-invariant ground states has been carried out in
\cite{Hak}.
\end{rem}

\section{Weakly Periodic Ground States}

In this section, we describe $H_A$-weakly periodic ground states. In
what follows, $H_A$ stands for the normal subgroup of index two. Due
to the definition of weakly periodic configuration, we infer that
each $H_A$-weakly periodic configuration has the following form:
\begin{equation}\label{gwper}
\v(x)=\left\{
    \begin{array}{llll}
        \s_1&if&\xd \in H_0,& x \in H_0,\\
        \s_2&if&\xd \in H_0,& x \in H_1,\\
        \s_3&if&\xd \in H_1,& x \in H_0,\\
        \s_4&if&\xd \in H_1,& x \in H_1,
    \end{array}\right.
\end{equation}
where $\s_i \in \F=\{1,2,3\}$, ($i=1, 2, 3, 4$).

Hence, all possible weakly periodic (non periodic) ground states
have the following forms:

\begin{minipage}{0.5\textwidth}
    \begin{equation*} 
    \v_1(x)=\left\{
    \begin{array}{llll}
    1&if&\xd \in H_0,& x \in H_0,\\
    1&if&\xd \in H_0,& x \in H_1,\\
    2&if&\xd \in H_1,& x \in H_0,\\
    2&if&\xd \in H_1,& x \in H_1,
    \end{array}\right.
    \end{equation*}
\end{minipage}
\begin{minipage}{0.5\textwidth}
    \begin{equation*} 
    \v_2(x)=\left\{
    \begin{array}{llll}
    1&if&\xd \in H_0,& x \in H_0,\\
    2&if&\xd \in H_0,& x \in H_1,\\
    2&if&\xd \in H_1,& x \in H_0,\\
    1&if&\xd \in H_1,& x \in H_1,
    \end{array}\right.
    \end{equation*}
\end{minipage}

\begin{minipage}{0.5\textwidth}
    \begin{equation*} 
    \v_3(x)=\left\{
    \begin{array}{llll}
    1&if&\xd \in H_0,& x \in H_0,\\
    2&if&\xd \in H_0,& x \in H_1,\\
    2&if&\xd \in H_1,& x \in H_0,\\
    3&if&\xd \in H_1,& x \in H_1,
    \end{array}\right.
    \end{equation*}
\end{minipage}
\begin{minipage}{0.5\textwidth}
    \begin{equation*} 
    \v_4(x)=\left\{
    \begin{array}{llll}
    2&if&\xd \in H_0,& x \in H_0,\\
    1&if&\xd \in H_0,& x \in H_1,\\
    1&if&\xd \in H_1,& x \in H_0,\\
    1&if&\xd \in H_1,& x \in H_1,
    \end{array}\right.
    \end{equation*}
\end{minipage}

\begin{minipage}{0.5\textwidth}
    \begin{equation*} 
    \v_5(x)=\left\{
    \begin{array}{llll}
    2&if&\xd \in H_0,& x \in H_0,\\
    1&if&\xd \in H_0,& x \in H_1,\\
    2&if&\xd \in H_1,& x \in H_0,\\
    2&if&\xd \in H_1,& x \in H_1,
    \end{array}\right.
    \end{equation*}
\end{minipage}
\begin{minipage}{0.5\textwidth}
    \begin{equation*} 
    \v_{6}(x)=\left\{
    \begin{array}{llll}
    2&if&\xd \in H_0,& x \in H_0,\\
    1&if&\xd \in H_0,& x \in H_1,\\
    3&if&\xd \in H_1,& x \in H_0,\\
    2&if&\xd \in H_1,& x \in H_1,
    \end{array}\right.
    \end{equation*}
\end{minipage}

\begin{minipage}{0.5\textwidth}
    \begin{equation*} 
    \v_{7}(x)=\left\{
    \begin{array}{llll}
    2&if&\xd \in H_0,& x \in H_0,\\
    2&if&\xd \in H_0,& x \in H_1,\\
    1&if&\xd \in H_1,& x \in H_0,\\
    2&if&\xd \in H_1,& x \in H_1,
    \end{array}\right.
    \end{equation*}
\end{minipage}
\begin{minipage}{0.5\textwidth}
    \begin{equation*} 
    \v_{8}(x)=\left\{
    \begin{array}{llll}
    2&if&\xd \in H_0,& x \in H_0,\\
    2&if&\xd \in H_0,& x \in H_1,\\
    2&if&\xd \in H_1,& x \in H_0,\\
    1&if&\xd \in H_1,& x \in H_1,
    \end{array}\right.
    \end{equation*}
\end{minipage}

\begin{minipage}{0.5\textwidth}
    \begin{equation*} 
    \v_{9}(x)=\left\{
    \begin{array}{llll}
    2&if&\xd \in H_0,& x \in H_0,\\
    2&if&\xd \in H_0,& x \in H_1,\\
    2&if&\xd \in H_1,& x \in H_0,\\
    3&if&\xd \in H_1,& x \in H_1,
    \end{array}\right.
    \end{equation*}
\end{minipage}
\begin{minipage}{0.5\textwidth}
    \begin{equation*} 
    \v_{10}(x)=\left\{
    \begin{array}{llll}
    2&if&\xd \in H_0,& x \in H_0,\\
    2&if&\xd \in H_0,& x \in H_1,\\
    3&if&\xd \in H_1,& x \in H_0,\\
    2&if&\xd \in H_1,& x \in H_1,
    \end{array}\right.
    \end{equation*}
\end{minipage}

\begin{minipage}{0.5\textwidth}
    \begin{equation*} 
    \v_{11}(x)=\left\{
    \begin{array}{llll}
    2&if&\xd \in H_0,& x \in H_0,\\
    2&if&\xd \in H_0,& x \in H_1,\\
    1&if&\xd \in H_1,& x \in H_0,\\
    1&if&\xd \in H_1,& x \in H_1,
    \end{array}\right.
    \end{equation*}
\end{minipage}
\begin{minipage}{0.5\textwidth}
    \begin{equation*} 
    \v_{12}(x)=\left\{
    \begin{array}{llll}
    3&if&\xd \in H_0,& x \in H_0,\\
    3&if&\xd \in H_0,& x \in H_1,\\
    3&if&\xd \in H_1,& x \in H_0,\\
    2&if&\xd \in H_1,& x \in H_1,
    \end{array}\right.
    \end{equation*}
\end{minipage}

\begin{minipage}{0.5\textwidth}
    \begin{equation*} 
    \v_{13}(x)=\left\{
    \begin{array}{llll}
    1&if&\xd \in H_0,& x \in H_0,\\
    1&if&\xd \in H_0,& x \in H_1,\\
    1&if&\xd \in H_1,& x \in H_0,\\
    2&if&\xd \in H_1,& x \in H_1,
    \end{array}\right.
    \end{equation*}
\end{minipage}
\begin{minipage}{0.5\textwidth}
    \begin{equation*} 
    \v_{14}(x)=\left\{
    \begin{array}{llll}
    1&if&\xd \in H_0,& x \in H_0,\\
    2&if&\xd \in H_0,& x \in H_1,\\
    2&if&\xd \in H_1,& x \in H_0,\\
    2&if&\xd \in H_1,& x \in H_1,
    \end{array}\right.
    \end{equation*}
\end{minipage}

\begin{minipage}{0.5\textwidth}
    \begin{equation*} 
    \v_{15}(x)=\left\{
    \begin{array}{llll}
    1&if&\xd \in H_0,& x \in H_0,\\
    3&if&\xd \in H_0,& x \in H_1,\\
    3&if&\xd \in H_1,& x \in H_0,\\
    3&if&\xd \in H_1,& x \in H_1,
    \end{array}\right.
    \end{equation*}
\end{minipage}
\begin{minipage}{0.5\textwidth}
    \begin{equation*} 
    \v_{16}(x)=\left\{
    \begin{array}{llll}
    3&if&\xd \in H_0,& x \in H_0,\\
    1&if&\xd \in H_0,& x \in H_1,\\
    1&if&\xd \in H_1,& x \in H_0,\\
    3&if&\xd \in H_1,& x \in H_1,
    \end{array}\right.
    \end{equation*}
\end{minipage}

\begin{minipage}{0.5\textwidth}
    \begin{equation*} 
    \v_{17}(x)=\left\{
    \begin{array}{llll}
    3&if&\xd \in H_0,& x \in H_0,\\
    3&if&\xd \in H_0,& x \in H_1,\\
    1&if&\xd \in H_1,& x \in H_0,\\
    1&if&\xd \in H_1,& x \in H_1,
    \end{array}\right.
    \end{equation*}
\end{minipage}
\begin{minipage}{0.5\textwidth}
    \begin{equation*} 
    \v_{18}(x)=\left\{
    \begin{array}{llll}
    3&if&\xd \in H_0,& x \in H_0,\\
    3&if&\xd \in H_0,& x \in H_1,\\
    1&if&\xd \in H_1,& x \in H_0,\\
    3&if&\xd \in H_1,& x \in H_1,
    \end{array}\right.
    \end{equation*}
\end{minipage}

\begin{minipage}{0.5\textwidth}
    \begin{equation*} 
    \v_{19}(x)=\left\{
    \begin{array}{llll}
    3&if&\xd \in H_0,& x \in H_0,\\
    3&if&\xd \in H_0,& x \in H_1,\\
    3&if&\xd \in H_1,& x \in H_0,\\
    1&if&\xd \in H_1,& x \in H_1,
    \end{array}\right.
    \end{equation*}
\end{minipage}
\begin{minipage}{0.5\textwidth}
    \begin{equation*} 
    \v_{20}(x)=\left\{
    \begin{array}{llll}
    1&if&\xd \in H_0,& x \in H_0,\\
    3&if&\xd \in H_0,& x \in H_1,\\
    3&if&\xd \in H_1,& x \in H_0,\\
    1&if&\xd \in H_1,& x \in H_1,
    \end{array}\right.
    \end{equation*}
\end{minipage}

\begin{thm}\label{th2}
    Let $k=2$ and $|A|=1$. Then for the $\l$-model given by \eqref{cond} the following statements hold:
    \begin{itemize}
        \item[I.] There exist fourteen $H_A$-weakly periodic ground states if $(\ab,\bb,\cb)\in A_9$ that are not $H_A$-periodic, and they are
        described by $\v_i$, $i=1,\dots,14$. Moreover, they are not ground states on the set $\mathbb{R}^3\setminus A_9$.
        \item[II.] There exist six $H_A$-weakly periodic ground states if $(\ab,\bb,\cb)\in A_5$  that are not $H_A$-periodic, and they are described by
        $\v_i$, $i=15,\dots,20$, and they are not ground states on the set $\mathbb{R}^3\setminus A_5$.
        \item[III.] All $H_A$-weakly periodic (non periodic, in particular non translation-invariant) configurations (not mentioned in assertions I and
        II) are not weakly periodic ground states.
    \end{itemize}
\end{thm}

\begin{proof}

    \begin{itemize}
        \item[1)] Consider
        \begin{equation*}
        \v_1(x)=\left\{
        \begin{array}{llll}
        1&if&\xd \in H_0,& x \in H_0,\\
        1&if&\xd \in H_0,& x \in H_1,\\
        2&if&\xd \in H_1,& x \in H_0,\\
        2&if&\xd \in H_1,& x \in H_1.
        \end{array}\right.
        \end{equation*}
        Let $\Cb \in H_0$, then we have the following possible cases
        \begin{itemize}
            \item[a)] $\Cbd \in H_0$ and $\v_{1,b}(\Cbd)=1$, then $\v_{1,b}(\Cb)=1,\ B^{(1)}=3, \ \v_{1,b} \in C_3$,
            \item[b)] $\Cbd \in H_0$ and $\v_{1,b}(\Cbd)=2$, then $\v_{1,b}(\Cb)=1,\ B^{(1)}=2,\ B^{(2)}=1, \ \v_{1,b} \in C_7$,
            \item[c)] $\Cbd \in H_1$ and $\v_{1,b}(\Cbd)=1$, this case is not exist, if $\Cbd\in H_1$, then $\v_{1,b}(\Cbd)=2$,
            \item[d)] $\Cbd \in H_1$ and $\v_{1,b}(\Cbd)=2$, then $\v_{1,b}(\Cb)=2,\ B^{(1)}=2,\ B^{(2)}=1, \ \v_{1,b} \in C_9$.

        \end{itemize}
        Let $\Cb \in H_1$, then the possible cases are
        \begin{itemize}
            \item[a)] $\Cbd \in H_0$ and $\v_{1,b}(\Cbd)=1$, then $\v_{1,b}(\Cb)=1,\ B^{(1)}=1,\ B^{(2)}=2, \ \v_{1,b} \in C_9$,
            \item[b)] $\Cbd \in H_0$ and $\v_{1,b}(\Cbd)=2$, this case is not exist, if $\Cbd\in H_0$, then $\v_{1,b}(\Cbd)=1$,
            \item[c)] $\Cbd \in H_1$ and $\v_{1,b}(\Cbd)=1$, then $\v_{1,b}(\Cb)=2,\ B^{(1)}=1,\ B^{(2)}=2,\ \v_{1,b} \in C_7$,
            \item[d)] $\Cbd \in H_1$ and $\v_{1,b}(\Cbd)=2$, then $\v_{1,b}(\Cb)=2,\ B^{(2)}=3,\ \v_{1,b} \in C_3$.
        \end{itemize}
        Therefore, the configuration $\v_1$ is a weakly periodic ground state on the set $A_3 \cap A_7 \cap A_9=A_7=A_9$.

        \item[2)] Consider
        \begin{equation*}
        \v_2(x)=\left\{
        \begin{array}{llll}
        1&if&\xd \in H_0,& x \in H_0,\\
        2&if&\xd \in H_0,& x \in H_1,\\
        2&if&\xd \in H_1,& x \in H_0,\\
        1&if&\xd \in H_1,& x \in H_1.
        \end{array}\right.
        \end{equation*}
        Let $\Cb \in H_0$, the possible cases are
        \begin{itemize}
            \item[a)] $\Cbd \in H_0$ and $\v_{2,b}(\Cbd)=1$, then $\v_{2,b}(\Cb)=2,\ B^{(1)}=2,\ B^{(2)}=1, \ \v_{2,b} \in C_7$,
            \item[b)] $\Cbd \in H_0$ and $\v_{2,b}(\Cbd)=2$, then $\v_{2,b}(\Cb)=2,\ B^{(1)}=1,\ B^{(2)}=2, \ \v_{2,b} \in C_9$,
            \item[c)] $\Cbd \in H_1$ and $\v_{2,b}(\Cbd)=1$, then $\v_{2,b}(\Cb)=2,\ B^{(1)}=3, \ \v_{2,b} \in C_2$,
            \item[d)] $\Cbd \in H_1$ and $\v_{2,b}(\Cbd)=2$, this case is not exist, if $\Cbd\in H_1$, then $\v_{2,b}(\Cbd)=1$.
        \end{itemize}
            Let $\Cb \in H_1$, the possible cases are
            \begin{itemize}
                \item[a)] $\Cbd \in H_0$ and $\v_{2,b}(\Cbd)=1$, then $\v_{2,b}(\Cb)=2,\ B^{(1)}=3, \ \v_{2,b} \in C_2$,
                \item[b)] $\Cbd \in H_0$ and $\v_{2,b}(\Cbd)=2$, this case is not exist, if $\Cbd\in H_0$, then $\v_{2,b}(\Cbd)=1$,
                \item[c)] $\Cbd \in H_1$ and $\v_{2,b}(\Cbd)=1$, then $\v_{2,b}(\Cb)=1,\ B^{(1)}=2,\ B^{(2)}=1,\ \v_{2,b}\in C_7$,
                \item[d)] $\Cbd \in H_1$ and $\v_{2,b}(\Cbd)=2$, then $\v_{2,b}(\Cb)=1,\ B^{(1)}=1,\ B^{(2)}=2,\ \v_{2,b}\in C_9$.
            \end{itemize}
            Therefore, configuration $\v_2$ a weakly periodic ground state on the set $A_2 \cap A_7 \cap A_9=A_7=A_9$.

        \item[3)] Consider
            \begin{equation*}
            \v_3(x)=\left\{
            \begin{array}{llll}
            1&if&\xd \in H_0,& x \in H_0,\\
            2&if&\xd \in H_0,& x \in H_1,\\
            2&if&\xd \in H_1,& x \in H_0,\\
            3&if&\xd \in H_1,& x \in H_1.
            \end{array}\right.
            \end{equation*}
            Let $\Cb \in H_0$, the possible cases are
            \begin{itemize}
                \item[a)] $\Cbd \in H_0$ and $\v_{3,b}(\Cbd)=1$, then $\v_{3,b}(\Cb)=1,\ B^{(1)}=2,\ B^{(2)}=1, \ \v_{3,b} \in C_7$,
                \item[b)] $\Cbd \in H_0$ and $\v_{3,b}(\Cbd)=2$, then $\v_{3,b}(\Cb)=2,\ B^{(1)}=1,\ B^{(2)}=2, \ \v_{3,b} \in C_9$,
                \item[c)] $\Cbd \in H_1$ and $\v_{3,b}(\Cbd)=2$, this case is not exist, if $\Cbd\in H_1$, then $\v_{3,b}(\Cbd)=3$,
                \item[d)] $\Cbd \in H_1$ and $\v_{3,b}(\Cbd)=3$, then $\v_{3,b}(\Cb)=2,\ B^{(1)}=2,\ B^{(3)}=1, \ \v_{3,b} \in C_9$.
            \end{itemize}
            Let $\Cb \in H_1$, the possible cases are
            \begin{itemize}
                \item[a)] $\Cbd \in H_0$ and $\v_{3,b}(\Cbd)=1$, then $\v_{3,b}(\Cb)=2,\ B^{(1)}=1,\ B^{(3)}=2, \ \v_{3,b} \in C_2$,
                \item[b)] $\Cbd \in H_0$ and $\v_{3,b}(\Cbd)=2$, this case is not exist, if $\Cbd\in H_0$, then $\v_{3,b}(\Cbd)=1$,
                \item[c)] $\Cbd \in H_1$ and $\v_{3,b}(\Cbd)=2$, then $\v_{3,b}(\Cb)=3,\ B^{(2)}=2,\ B^{(3)}=1,\ \v_{3,b} \in C_9$,
                \item[d)] $\Cbd \in H_1$ and $\v_{3,b}(\Cbd)=3$, then $\v_{3,b}(\Cb)=3,\ B^{(2)}=1,\ B^{(3)}=2,\ \v_{3,b}\in C_9$.
            \end{itemize}
            Therefore, configuration $\v_3$ a weakly periodic ground state on the set $A_2 \cap A_7 \cap A_9=A_7=A_9$.

        \item[4)] consider
        \begin{equation*}
        \v_4(x)=\left\{
        \begin{array}{llll}
        2&if&\xd \in H_0,& x \in H_0,\\
        1&if&\xd \in H_0,& x \in H_1,\\
        1&if&\xd \in H_1,& x \in H_0,\\
        1&if&\xd \in H_1,& x \in H_1.
        \end{array}\right.
        \end{equation*}
        Let $\Cb \in H_0$, the possible cases are
        \begin{itemize}
            \item[a)] $\Cbd \in H_0$ and $\v_{4,b}(\Cbd)=1$, then $\v_{4,b}(\Cb)=2,\ B^{(1)}=2,\ B^{(2)}=1, \ \v_{4,b} \in C_9$,
            \item[b)] $\Cbd \in H_0$ and $\v_{4,b}(\Cbd)=2$, then $\v_{4,b}(\Cb)=2,\ B^{(1)}=1,\ B^{(2)}=2, \ \v_{4,b} \in C_7$,
            \item[c)] $\Cbd \in H_1$ and $\v_{4,b}(\Cbd)=1$, then $\v_{4,b}(\Cb)=1,\ B^{(1)}=1,\ B^{(2)}=2, \ \v_{4,b} \in C_9$.
        \end{itemize}
        Let $\Cb \in H_1$, the possible cases are
        \begin{itemize}
            \item[a)] $\Cbd \in H_0$ and $\v_{4,b}(\Cbd)=1$, then $\v_{4,b}(\Cb)=1,\ B^{(1)}=3, \ \v_{4,b} \in C_3$,
            \item[b)] $\Cbd \in H_0$ and $\v_{4,b}(\Cbd)=2$, then $\v_{4,b}(\Cb)=1,\ B^{(1)}=2,\ B^{(2)}=1, \ \v_{4,b} \in C_7$,
            \item[c)] $\Cbd \in H_1$ and $\v_{4,b}(\Cbd)=1$, then $\v_{4,b}(\Cb)=1,\ B^{(1)}=3, \ \v_{4,b} \in C_9$.
        \end{itemize}
        Therefore, configuration $\v_4$ a weakly periodic ground state on the set $A_3 \cap A_7 \cap A_9=A_7=A_9$.

        \item[5)] Consider
        \begin{equation*}
        \v_5(x)=\left\{
        \begin{array}{llll}
        2&if&\xd \in H_0,& x \in H_0,\\
        1&if&\xd \in H_0,& x \in H_1,\\
        2&if&\xd \in H_1,& x \in H_0,\\
        2&if&\xd \in H_1,& x \in H_1.
        \end{array}\right.
        \end{equation*}
        Let $\Cb \in H_0$, the possible cases are
        \begin{itemize}
            \item[a)] $\Cbd \in H_0$ and $\v_{5,b}(\Cbd)=2$, then $\v_{5,b}(\Cb)=2,\ B^{(1)}=1,\ B^{(2)}=2, \ \v_{5,b} \in C_7$,
            \item[b)] $\Cbd \in H_1$ and $\v_{5,b}(\Cbd)=1$, this case is not exist, if $\Cbd\in H_1$, then $\v_{5,b}(\Cbd)=2$,
            \item[c)] $\Cbd \in H_1$ and $\v_{5,b}(\Cbd)=2$, then $\v_{5,b}(\Cb)=2,\ B^{(2)}=3, \ \v_{5,b} \in C_3$.
        \end{itemize}
        Let $\Cb \in H_1$, the possible cases are
        \begin{itemize}
            \item[a)] $\Cbd \in H_0$ and $\v_{5,b}(\Cbd)=2$, then $\v_{5,b}(\Cb)=1,\ B^{(2)}=3, \ \v_{5,b} \in C_2$,
            \item[b)] $\Cbd \in H_1$ and $\v_{5,b}(\Cbd)=1$, then $\v_{5,b}(\Cb)=2,\ B^{(1)}=2,\ B^{(2)}=1, \ \v_{5,b} \in C_9$,
            \item[c)] $\Cbd \in H_1$ and $\v_{5,b}(\Cbd)=2$, then $\v_{5,b}(\Cb)=2,\ B^{(2)}=3, \ \v_{5,b}\in C_3$.
        \end{itemize}
        Therefore, configuration $\v_5$ a weakly periodic ground state on the set $A_2 \cap A_3 \cap A_7 \cap A_9=A_7=A_9$.
        \item[6)] Consider
        \begin{equation*}
        \v_{6}(x)=\left\{
        \begin{array}{llll}
        2&if&\xd \in H_0,& x \in H_0,\\
        1&if&\xd \in H_0,& x \in H_1,\\
        3&if&\xd \in H_1,& x \in H_0,\\
        2&if&\xd \in H_1,& x \in H_1.
        \end{array}\right.
        \end{equation*}
        Let $\Cb \in H_0$, the possible cases are
        \begin{itemize}
            \item[a)] $\Cbd \in H_0$ and $\v_{6,b}(\Cbd)=2$, then $\v_{6,b}(\Cb)=2,\ B^{(2)}=2,\ B^{(3)}=3, \ \v_{6,b} \in C_7$,
            \item[b)] $\Cbd \in H_0$ and $\v_{6,b}(\Cbd)=3$, then $\v_{6,b}(\Cb)=2,\ B^{(1)}=B^{(2)}=B^{(3)}=1, \ \v_{6,b} \in C_9$,
            \item[c)] $\Cbd \in H_1$ and $\v_{6,b}(\Cbd)=1$, this case is not exist, if $\Cbd\in H_1$, then $\v_{6,b}(\Cbd)=2$,
            \item[d)] $\Cbd \in H_1$ and $\v_{6,b}(\Cbd)=2$, then $\v_{6,b}(\Cb)=3,\ B^{(2)}=3, \ \v_{6,b} \in C_2$.

        \end{itemize}
        Let $\Cb \in H_1$, the possible cases are
        \begin{itemize}
            \item[a)] $\Cbd \in H_0$ and $\v_{6,b}(\Cbd)=2$, then $\v_{6,b}(\Cb)=1,\ B^{(2)}=3, \ \v_{6,b} \in C_2$,
            \item[b)] $\Cbd \in H_0$ and $\v_{6,b}(\Cbd)=3$, this case is not exist, if $\Cbd\in H_0$, then $\v_{6,b}(\Cbd)=2$,
            \item[c)] $\Cbd \in H_1$ and $\v_{6,b}(\Cbd)=1$, then $\v_{6,b}(\Cb)=2,\ B^{(1)}=B^{(2)}=B^{(3)}=1,\ \v_{6,b} \in C_9$,
            \item[d)] $\Cbd \in H_1$ and $\v_{6,b}(\Cbd)=2$, then $\v_{6,b}(\Cb)=2,\ B^{(2)}=2,\ B^{(3)}=1,\ \v_{6,b} \in C_7$.
        \end{itemize}
        Therefore, configuration $\v_{6}$ a weakly periodic ground state on the set $A_2 \cap A_7 \cap A_9=A_7=A_9$.

        \item[7)] Consider
        \begin{equation*}
        \v_{7}(x)=\left\{
        \begin{array}{llll}
        2&if&\xd \in H_0,& x \in H_0,\\
        2&if&\xd \in H_0,& x \in H_1,\\
        1&if&\xd \in H_1,& x \in H_0,\\
        2&if&\xd \in H_1,& x \in H_1.
        \end{array}\right.
        \end{equation*}
        Let $\Cb \in H_0$, the possible cases are
        \begin{itemize}
            \item[a)] $\Cbd \in H_0$ and $\v_{7,b}(\Cbd)=1$, then $\v_{7,b}(\Cb)=2,\ B^{(1)}=1,\ B^{(2)}=2, \ \v_{7,b} \in C_7$,
            \item[b)] $\Cbd \in H_0$ and $\v_{7,b}(\Cbd)=2$, then $\v_{7,b}(\Cb)=2,\ B^{(2)}=3, \ \v_{7,b} \in C_3$,
            \item[c)] $\Cbd \in H_1$ and $\v_{7,b}(\Cbd)=2$, then $\v_{7,b}(\Cb)=1,\ B^{(2)}=3, \ \v_{7,b} \in C_2$.
        \end{itemize}
        Let $\Cb \in H_1$, the possible cases are
        \begin{itemize}
            \item[a)] $\Cbd \in H_0$ and $\v_{7,b}(\Cbd)=1$, then $\v_{7,b}(\Cb)=2,\ B^{(1)}=1,\ B^{(2)}=2, \ \v_{7,b} \in C_7$,
            \item[b)] $\Cbd \in H_0$ and $\v_{7,b}(\Cbd)=2$, then $\v_{7,b}(\Cb)=2,\ B^{(2)}=3, \ \v_{7,b} \in C_3$,
            \item[c)] $\Cbd \in H_1$ and $\v_{7,b}(\Cbd)=2$, then $\v_{7,b}(\Cb)=2,\ B^{(1)}=1,\ B^{(2)}=2, \ \v_{7,b} \in C_7$.
        \end{itemize}
        Therefore, configuration $\v_{7}$ a weakly periodic ground state on the set $A_2 \cap A_3 \cap A_7=A_7$.

        \item[8)] Consider
        \begin{equation*}
        \v_{8}(x)=\left\{
        \begin{array}{llll}
        2&if&\xd \in H_0,& x \in H_0,\\
        2&if&\xd \in H_0,& x \in H_1,\\
        2&if&\xd \in H_1,& x \in H_0,\\
        1&if&\xd \in H_1,& x \in H_1.
        \end{array}\right.
        \end{equation*}
        Let $\Cb \in H_0$, the possible cases are
        \begin{itemize}
            \item[a)] $\Cbd \in H_0$ and $\v_{8,b}(\Cbd)=2$, then $\v_{8,b}(\Cb)=2,\ B^{(2)}=3, \ \v_{8,b} \in C_3$,
            \item[b)] $\Cbd \in H_1$ and $\v_{8,b}(\Cbd)=1$, then $\v_{8,b}(\Cb)=2,\ B^{(1)}=1,\ B^{(2)}=2, \ \v_{8,b} \in C_7$,
            \item[c)] $\Cbd \in H_1$ and $\v_{8,b}(\Cbd)=2$, this case is not exist, if $\Cbd\in H_1$, then $\v_{8,b}(\Cbd)=1$.
        \end{itemize}
        Let $\Cb \in H_1$, the possible cases are
        \begin{itemize}
            \item[a)] $\Cbd \in H_0$ and $\v_{8,b}(\Cbd)=2$, then we can have two configuration:
            \begin{itemize}
                \item[(i)] $\v_{8,b}(\Cb)=2,\ B^{(1)}=1,\ B^{(2)}=2, \ \v_{8,b} \in C_7$,
                \item[(ii)] $\v_{8,b}(\Cb)=2,\ B^{(1)}=3, \ \v_{8,b} \in C_1$,
            \end{itemize}
            \item[b)] $\Cbd \in H_1$ and $\v_{8,b}(\Cbd)=1$, then $\v_{8,b}(\Cb)=1,\ B^{(1)}=2,\ B^{(2)}=1, \ \v_{8,b} \in C_7$,
            \item[c)] $\Cbd \in H_1$ and $\v_{8,b}(\Cbd)=2$, then $\v_{8,b}(\Cb)=1,\ B^{(1)}=1,\ B^{(2)}=2, \ \v_{8,b} \in C_9$.
        \end{itemize}
        Therefore, configuration $\v_{8}$ a weakly periodic ground state on the set $A_2 \cap A_3 \cap A_7 \cap A_9=A_7=A_9$.

        \item[9)] Consider
        \begin{equation*}
        \v_{9}(x)=\left\{
        \begin{array}{llll}
        2&if&\xd \in H_0,& x \in H_0,\\
        2&if&\xd \in H_0,& x \in H_1,\\
        2&if&\xd \in H_1,& x \in H_0,\\
        3&if&\xd \in H_1,& x \in H_1.
        \end{array}\right.
        \end{equation*}
        Let $\Cb \in H_0$, the possible cases are
        \begin{itemize}
            \item[a)] $\Cbd \in H_0$ and $\v_{9,b}(\Cbd)=2$, then $\v_{9,b}(\Cb)=2,\ B^{(2)}=3, \ \v_{9,b} \in C_3$,
            \item[b)] $\Cbd \in H_1$ and $\v_{9,b}(\Cbd)=2$, this case is not exist, if $\Cbd\in H_1$, then $\v_{9,b}(\Cbd)=3$,
            \item[c)] $\Cbd \in H_1$ and $\v_{9,b}(\Cbd)=3$, then $\v_{9,b}(\Cb)=2,\ B^{(2)}=2,\ B^{(3)}=1, \ \v_{9,b} \in C_7$.
        \end{itemize}
        Let $\Cb \in H_1$, the possible cases are
        \begin{itemize}
            \item[a)] $\Cbd \in H_0$ and $\v_{9,b}(\Cbd)=2$, then $\v_{9,b}(\Cb)=2,\ B^{(2)}=1,\ B^{(3)}=2, \ \v_{9,b} \in C_9$,
            \item[b)] $\Cbd \in H_1$ and $\v_{9,b}(\Cbd)=2$, then $\v_{9,b}(\Cb)=3,\ B^{(2)}=2,\ B^{(3)}=1, \ \v_{9,b} \in C_9$,
            \item[c)] $\Cbd \in H_1$ and $\v_{9,b}(\Cbd)=3$, then $\v_{9,b}(\Cb)=3,\ B^{(2)}=1,\ B^{(3)}=2, \ \v_{9,b} \in C_7$.
        \end{itemize}
        Therefore, configuration $\v_{9}$ a weakly periodic ground state on the set $A_3 \cap A_7 \cap A_9=A_7=A_9$.

        \item[10)] Consider
        \begin{equation*}
        \v_{10}(x)=\left\{
        \begin{array}{llll}
        2&if&\xd \in H_0,& x \in H_0,\\
        2&if&\xd \in H_0,& x \in H_1,\\
        3&if&\xd \in H_1,& x \in H_0,\\
        2&if&\xd \in H_1,& x \in H_1.
        \end{array}\right.
        \end{equation*}
        Let $\Cb \in H_0$, the possible cases are
        \begin{itemize}
            \item[a)] $\Cbd \in H_0$ and $\v_{10,b}(\Cbd)=2$, then $\v_{10,b}(\Cb)=2,\ B^{(2)}=3, \ \v_{10,b} \in C_3$,
            \item[b)] $\Cbd \in H_0$ and $\v_{10,b}(\Cbd)=3$, then $\v_{10,b}(\Cb)=2,\ B^{(2)}=2,\ B^{(3)}=1, \ \v_{10,b} \in C_7$,
            \item[c)] $\Cbd \in H_1$ and $\v_{10,b}(\Cbd)=2$, then $\v_{10,b}(\Cb)=2,\ B^{(2)}=3, \ \v_{10,b} \in C_7$.
        \end{itemize}
        Let $\Cb \in H_1$, the possible cases are
        \begin{itemize}
            \item[a)] $\Cbd \in H_0$ and $\v_{10,b}(\Cbd)=2$, then $\v_{10,b}(\Cb)=2,\ B^{(2)}=3, \ \v_{10,b}\in C_3$,
            \item[b)] $\Cbd \in H_0$ and $\v_{10,b}(\Cbd)=3$, this case is not exist, if $\Cbd\in H_0$, then $\v_{10,b}(\Cbd)=2$,
            \item[c)] $\Cbd \in H_1$ and $\v_{10,b}(\Cbd)=3$, then $\v_{10,b}(\Cb)=2,\ B^{(2)}=2,\ B^{(3)}=1, \ \v_{10,b} \in C_7$.
        \end{itemize}
        Therefore, configuration $\v_{10}$ a weakly periodic ground state on the set $A_2 \cap A_3 \cap A_7=A_7$.

    \item[11)] Consider
        \begin{equation*}
        \v_{11}(x)=\left\{
        \begin{array}{llll}
        2&if&\xd \in H_0,& x \in H_0,\\
        2&if&\xd \in H_0,& x \in H_1,\\
        1&if&\xd \in H_1,& x \in H_0,\\
        1&if&\xd \in H_1,& x \in H_1.
        \end{array}\right.
        \end{equation*}
            Let $\Cb \in H_0$, the possible cases are
            \begin{itemize}
                \item[a)] $\Cbd \in H_0$ and $\v_{11,b}(\Cbd)=1$, then $\v_{11,b}(\Cb)=2,\ B^{(1)}=1,\ B^{(2)}=2, \ \v_{11,b} \in C_7$,
                \item[b)] $\Cbd \in H_0$ and $\v_{11,b}(\Cbd)=2$, then $\v_{11,b}(\Cb)=2,\ B^{(2)}=3, \ \v_{11,b} \in C_3$,
                \item[c)] $\Cbd \in H_1$ and $\v_{11,b}(\Cbd)=1$, then $\v_{11,b}(\Cb)=1,\ B^{(1)}=1,\ B^{(2)}=2 \ \v_{11,b} \in C_9$,
                \item[d)] $\Cbd \in H_1$ and $\v_{11,b}(\Cbd)=2$, then $\v_{11,b}(\Cb)=1,\ B^{(1)}=2,\ B^{(2)}=1 \ \v_{11,b} \in C_7$.
            \end{itemize}
            Let $\Cb \in H_1$, the possible cases are
            \begin{itemize}
                \item[a)] $\Cbd \in H_0$ and $\v_{11,b}(\Cbd)=1$, then $\v_{11,b}(\Cb)=2,\ B^{(1)}=3, \ \v_{11,b} \in C_2$,
                \item[b)] $\Cbd \in H_0$ and $\v_{11,b}(\Cbd)=2$, then $\v_{11,b}(\Cb)=2,\ B^{(1)}=2,\ B^{(2)}=1,\ \v_{11,b} \in C_9$,
                \item[c)] $\Cbd \in H_1$ and $\v_{11,b}(\Cbd)=1$, then $\v_{11,b}(\Cb)=1,\ B^{(1)}=3,\ \v_{11,b}\in C_3$,
                \item[d)] $\Cbd \in H_1$ and $\v_{11,b}(\Cbd)=2$, then $\v_{11,b}(\Cb)=1,\ B^{(1)}=2,\ B^{(2)}=1,\ \v_{11,b} \in C_7$.
            \end{itemize}
            Therefore, configuration $\v_{11}$ a weakly periodic ground state on the set $A_2 \cap A_3 \cap A_7 \cap A_9=A_7=A_9$.

            \item[12)] Consider
            \begin{equation*}
            \v_{12}(x)=\left\{
            \begin{array}{llll}
            3&if&\xd \in H_0,& x \in H_0,\\
            3&if&\xd \in H_0,& x \in H_1,\\
            3&if&\xd \in H_1,& x \in H_0,\\
            2&if&\xd \in H_1,& x \in H_1.
            \end{array}\right.
            \end{equation*}
            Let $\Cb \in H_0$, the possible cases are
            \begin{itemize}
                \item[a)] $\Cbd \in H_0$ and $\v_{12,b}(\Cbd)=3$, then $\v_{12,b}(\Cb)=3,\ B^{(3)}=3,\ \v_{12,b} \in C_3$,
                \item[b)] $\Cbd \in H_1$ and $\v_{12,b}(\Cbd)=2$, then $\v_{12,b}(\Cb)=3,\ B^{(2)}=1,\ B^{(3)}=2, \ \v_{12,b} \in C_7$,
                \item[c)] $\Cbd \in H_1$ and $\v_{12,b}(\Cbd)=3$, this case is not exist, if $\Cbd\in H_1$, then $\v_{12,b}(\Cbd)=2$.
            \end{itemize}
            Let $\Cb \in H_1$, the possible cases are
            \begin{itemize}
                \item[a)] $\Cbd \in H_0$ and $\v_{12,b}(\Cbd)=3$, then $\v_{12,b}(\Cb)=3,\ B^{(2)}=2,\ B^{(3)}=1, \ \v_{12,b} \in C_9$,
                \item[b)] $\Cbd \in H_1$ and $\v_{12,b}(\Cbd)=2$, then $\v_{12,b}(\Cb)=2,\ B^{(2)}=2, \ B^{(3)}=1, \ \v_{12,b} \in C_7$,
                \item[c)] $\Cbd \in H_1$ and $\v_{12,b}(\Cbd)=3$, then $\v_{12,b}(\Cb)=2,\ B^{(2)}=1,\ B^{(3)}=2, \ \v_{12,b}\in C_9$.
            \end{itemize}
            Therefore, configuration $\v_{12}$ a weakly periodic ground state on the set $A_1 \cap A_3 \cap A_7 \cap A_9=A_7=A_9$.

            \item[13)] Consider
            \begin{equation*}
            \v_{13}(x)=\left\{
            \begin{array}{llll}
            1&if&\xd \in H_0,& x \in H_0,\\
            1&if&\xd \in H_0,& x \in H_1,\\
            1&if&\xd \in H_1,& x \in H_0,\\
            2&if&\xd \in H_1,& x \in H_1.
            \end{array}\right.
            \end{equation*}
            Let $\Cb \in H_0$, the possible cases are
            \begin{itemize}
                \item[a)] $\Cbd \in H_0$ and $\v_{13,b}(\Cbd)=1$, then $\v_{13,b}(\Cb)=1,\ B^{(1)}=3,\ \v_{13,b} \in C_3$,
                \item[b)] $\Cbd \in H_1$ and $\v_{13,b}(\Cbd)=1$, this case is not exist, if $\Cbd\in H_1$, then $\v_{13,b}(\Cbd)=2$,
                \item[c)] $\Cbd \in H_1$ and $\v_{13,b}(\Cbd)=2$, then $\v_{3,b}(\Cb)=2,\ B^{(1)}=2,\ B^{(2)}=1, \ \v_{13,b} \in C_7$.
            \end{itemize}
            Let $\Cb \in H_1$, the possible cases are
            \begin{itemize}
                \item[a)] $\Cbd \in H_0$ and $\v_{13,b}(\Cbd)=1$, then $\v_{13,b}(\Cb)=1,\ B^{(1)}=1,\ B^{(2)}=2, \ \v_{13,b} \in C_9$,
                \item[b)] $\Cbd \in H_1$ and $\v_{13,b}(\Cbd)=1$, then $\v_{13,b}(\Cb)=2,\ B^{(1)}=2, \ B^{(2)}=1, \ \v_{13,b} \in C_9$,
                \item[c)] $\Cbd \in H_1$ and $\v_{13,b}(\Cbd)=2$, then $\v_{13,b}(\Cb)=2,\ B^{(1)}=1,\ B^{(2)}=2, \ \v_{13,b} \in C_7$.
            \end{itemize}
            Therefore, configuration $\v_{13}$ a weakly periodic ground state on the set $A_3 \cap A_7 \cap A_9=A_7=A_9$.

            \item[14)] Consider
            \begin{equation*}
            \v_{14}(x)=\left\{
            \begin{array}{llll}
            1&if&\xd \in H_0,& x \in H_0,\\
            2&if&\xd \in H_0,& x \in H_1,\\
            2&if&\xd \in H_1,& x \in H_0,\\
            2&if&\xd \in H_1,& x \in H_1.
            \end{array}\right.
            \end{equation*}
            Let $\Cb \in H_0$, the possible cases are
            \begin{itemize}
                \item[a)] $\Cbd \in H_0$ and $\v_{14,b}(\Cbd)=1$, then $\v_{14,b}(\Cb)=1,\ B^{(1)}=2,\ B^{(2)}=1, \ \v_{14,b} \in C_7$,
                \item[b)] $\Cbd \in H_0$ and $\v_{14,b}(\Cbd)=2$, then $\v_{14,b}(\Cb)=1,\ B^{(1)}=1, \ B^{(2)}=2, \ \v_{14,b} \in C_9$,
                \item[c)] $\Cbd \in H_1$ and $\v_{14,b}(\Cbd)=2$, then $\v_{14,b}(\Cb)=1,\ B^{(1)}=1,\ B^{(2)}=2, \ \v_{14,b}\in C_9$.
            \end{itemize}
            Let $\Cb \in H_1$, the possible cases are
            \begin{itemize}
                \item[a)] $\Cbd \in H_0$ and $\v_{14,b}(\Cbd)=1$, then $\v_{14,b}(\Cb)=2,\ B^{(1)}=1,\ B^{(2)}=2,\ \v_{14,b} \in C_7$,
                \item[b)] $\Cbd \in H_1$ and $\v_{14,b}(\Cbd)=2$, this case is not exist, if $\Cbd\in H_1$, then $\v_{14,b}(\Cbd)=1$,
                \item[c)] $\Cbd \in H_1$ and $\v_{14,b}(\Cbd)=2$, then $\v_{14,b}(\Cb)=2,\ B^{(2)}=3,\ \v_{14,b} \in C_3$.
            \end{itemize}
            Therefore, configuration $\v_{14}$ a weakly periodic ground state on the set $A_3 \cap A_7 \cap A_9=A_7=A_9$.

    \item[15)] Consider
        \begin{equation*}
        \v_6(x)=\left\{
        \begin{array}{llll}
        1&if&\xd \in H_0,& x \in H_0,\\
        3&if&\xd \in H_0,& x \in H_1,\\
        3&if&\xd \in H_1,& x \in H_0,\\
        3&if&\xd \in H_1,& x \in H_1.
        \end{array}\right.
        \end{equation*}
        Let $\Cb \in H_0$, the possible cases are
        \begin{itemize}
            \item[a)] $\Cbd \in H_0$ and $\v_{15,b}(\Cbd)=1$, then $\v_{15,b}(\Cb)=1,\ B^{(1)}=2,\ B^{(3)}=1, \ \v_{15,b} \in C_5$,
            \item[b)] $\Cbd \in H_0$ and $\v_{15,b}(\Cbd)=3$, then $\v_{15,b}(\Cb)=1,\ B^{(1)}=1,\ B^{(3)}=2, \ \v_{15,b} \in C_8$,
            \item[c)] $\Cbd \in H_1$ and $\v_{15,b}(\Cbd)=3$, then $\v_{15,b}(\Cb)=3,\ B^{(1)}=2,\ B^{(3)}=1, \ \v_{15,b} \in C_8$.
        \end{itemize}
        Let $\Cb \in H_1$, the possible cases are
        \begin{itemize}
            \item[a)] $\Cbd \in H_0$ and $\v_{15,b}(\Cbd)=1$, then $\v_{15,b}(\Cb)=3,\ B^{(1)}=1,\ B^{(3)}=2, \ \v_{15,b} \in C_5$,
            \item[b)] $\Cbd \in H_0$ and $\v_{15,b}(\Cbd)=3$, this case is not exist, if $\Cbd\in H_0$, then $\v_{15,b}(\Cbd)=1$,
            \item[c)] $\Cbd \in H_1$ and $\v_{15,b}(\Cbd)=3$, then $\v_{15,b}(\Cb)=3,\ B^{(3)}=3,\ \v_{15,b}\in C_3$.
        \end{itemize}
        Therefore, configuration $\v_{15}$ a weakly periodic ground state on the set $A_3 \cap A_5 \cap A_8=A_5=A_8$.

    \item[16)] Consider
            \begin{equation*}
            \v_{16}(x)=\left\{
            \begin{array}{llll}
            3&if&\xd \in H_0,& x \in H_0,\\
            1&if&\xd \in H_0,& x \in H_1,\\
            1&if&\xd \in H_1,& x \in H_0,\\
            3&if&\xd \in H_1,& x \in H_1.
            \end{array}\right.
            \end{equation*}
            Let $\Cb \in H_0$, the possible cases are
            \begin{itemize}
                \item[a)] $\Cbd \in H_0$ and $\v_{16,b}(\Cbd)=1$, then $\v_{16,b}(\Cb)=3,\ B^{(1)}=2,\ B^{(3)}=1, \ \v_{16,b} \in C_8$,
                \item[b)] $\Cbd \in H_0$ and $\v_{16,b}(\Cbd)=3$, then $\v_{16,b}(\Cb)=3,\ B^{(1)}=1,\ B^{(3)}=2, \ \v_{16,b} \in C_5$,
                \item[c)] $\Cbd \in H_1$ and $\v_{16,b}(\Cbd)=1$, this case is not exist, if $\Cbd\in H_1$, then $\v_{16,b}(\Cbd)=3$,
                \item[d)] $\Cbd \in H_1$ and $\v_{16,b}(\Cbd)=3$, then $\v_{16,b}(\Cb)=1,\ B^{(3)}=3,\ \v_{16,b} \in C_1$.
            \end{itemize}
            Let $\Cb \in H_1$, the possible cases are
            \begin{itemize}
                \item[a)] $\Cbd \in H_0$ and $\v_{16,b}(\Cbd)=1$, then $\v_{16,b}(\Cb)=1,\ B^{(1)}=1,\ B^{(3)}=2, \ \v_{16,b} \in C_8$,
                \item[b)] $\Cbd \in H_0$ and $\v_{16,b}(\Cbd)=3$, then $\v_{16,b}(\Cb)=1,\ B^{(3)}=3,\ \v_{16,b} \in C_2$,
                \item[c)] $\Cbd \in H_1$ and $\v_{16,b}(\Cbd)=1$, then $\v_{16,b}(\Cb)=3,\ B^{(1)}=2,\ B^{(3)}=1,\ \v_{16,b} \in C_8$,
                \item[d)] $\Cbd \in H_1$ and $\v_{16,b}(\Cbd)=3$, then $\v_{16,b}(\Cb)=3,\ B^{(1)}=1,\ B^{(3)}=2,\ \v_{16,b} \in C_5$.
            \end{itemize}
            Therefore, configuration $\v_{16}$ a weakly periodic ground state on the set $A_1 \cap A_5 \cap A_8=A_5=A_8$.

            \item[17)] Consider
            \begin{equation*}
            \v_{17}(x)=\left\{
            \begin{array}{llll}
            3&if&\xd \in H_0,& x \in H_0,\\
            3&if&\xd \in H_0,& x \in H_1,\\
            1&if&\xd \in H_1,& x \in H_0,\\
            1&if&\xd \in H_1,& x \in H_1.
            \end{array}\right.
            \end{equation*}
            Let $\Cb \in H_0$, the possible cases are
            \begin{itemize}
                \item[a)] $\Cbd \in H_0$ and $\v_{17,b}(\Cbd)=1$, then $\v_{17,b}(\Cb)=3,\ B^{(1)}=1,\ B^{(3)}=2, \ \v_{17,b} \in C_5$,
                \item[b)] $\Cbd \in H_0$ and $\v_{17,b}(\Cbd)=3$, then $\v_{17,b}(\Cb)=3,\ B^{(3)}=3, \ \v_{17,b} \in C_3$,
                \item[c)] $\Cbd \in H_1$ and $\v_{17,b}(\Cbd)=1$, then $\v_{17,b}(\Cb)=1,\ B^{(1)}=1,\ B^{(3)}=2,\ \v_{17,b} \in C_8$,
                \item[d)] $\Cbd \in H_1$ and $\v_{17,b}(\Cbd)=3$, this case is not exist, if $\Cbd\in H_1$, then $\v_{17,b}(\Cbd)=1$.
            \end{itemize}
            Let $\Cb \in H_1$, the possible cases are
            \begin{itemize}
                \item[a)] $\Cbd \in H_0$ and $\v_{17,b}(\Cbd)=1$, then $\v_{17,b}(\Cb)=3,\ B^{(1)}=3,\ \v_{17,b} \in C_1$,
                \item[b)] $\Cbd \in H_0$ and $\v_{17,b}(\Cbd)=3$, then $\v_{17,b}(\Cb)=3,\ B^{(1)}=2,\ B^{(3)}=1,\ \v_{17,b} \in C_8$,
                \item[c)] $\Cbd \in H_1$ and $\v_{17,b}(\Cbd)=1$, then $\v_{17,b}(\Cb)=1,\ B^{(1)}=3,\ \v_{17,b} \in C_3$,
                \item[d)] $\Cbd \in H_1$ and $\v_{17,b}(\Cbd)=3$, then $\v_{17,b}(\Cb)=1,\ B^{(1)}=2,\ B^{(3)}=1,\ \v_{17,b} \in C_5$.
            \end{itemize}
            Therefore, configuration $\v_{17}$ a weakly periodic ground state on the set $A_1 \cap A_3 \cap A_5 \cap A_8=A_5=A_8$.

            \item[18)] Consider
            \begin{equation*}
            \v_{18}(x)=\left\{
            \begin{array}{llll}
            3&if&\xd \in H_0,& x \in H_0,\\
            3&if&\xd \in H_0,& x \in H_1,\\
            1&if&\xd \in H_1,& x \in H_0,\\
            3&if&\xd \in H_1,& x \in H_1.
            \end{array}\right.
            \end{equation*}
            Let $\Cb \in H_0$, the possible cases are
            \begin{itemize}
                \item[a)] $\Cbd \in H_0$ and $\v_{18,b}(\Cbd)=1$, then $\v_{18,b}(\Cb)=3,\ B^{(1)}=1,\ B^{(3)}=2, \ \v_{18,b} \in C_5$,
                \item[b)] $\Cbd \in H_0$ and $\v_{18,b}(\Cbd)=3$, then $\v_{18,b}(\Cb)=3,\ B^{(3)}=3,\ \v_{18,b}\in C_3$,
                \item[c)] $\Cbd \in H_1$ and $\v_{18,b}(\Cbd)=3$, then $\v_{18,b}(\Cb)=1,\ B^{(3)}=3, \ \v_{18,b}\in C_1$.
            \end{itemize}
            Let $\Cb \in H_1$, the possible cases are
            \begin{itemize}
                \item[a)] $\Cbd \in H_0$ and $\v_{18,b}(\Cbd)=1$, then $\v_{18,b}(\Cb)=3,\ B^{(1)}=1,\ B^{(3)}=2, \ \v_{18,b} \in C_5$,
                \item[b)] $\Cbd \in H_0$ and $\v_{18,b}(\Cbd)=3$, then $\v_{18,b}(\Cb)=3,\ B^{(3)}=3, \ \v_{18,b} \in C_3$,
                \item[c)] $\Cbd \in H_1$ and $\v_{18,b}(\Cbd)=3$, then $\v_{18,b}(\Cb)=3,\ B^{(1)}=1,\ B^{(3)}=2, \ \v_{18,b} \in C_5$.
            \end{itemize}
            Therefore, configuration $\v_{18}$ a weakly periodic ground state on the set $A_1 \cap A_3 \cap A_5=A_5$.

            \item[19)] Consider
            \begin{equation*}
            \v_{19}(x)=\left\{
            \begin{array}{llll}
            3&if&\xd \in H_0,& x \in H_0,\\
            3&if&\xd \in H_0,& x \in H_1,\\
            3&if&\xd \in H_1,& x \in H_0,\\
            1&if&\xd \in H_1,& x \in H_1.
            \end{array}\right.
            \end{equation*}
            Let $\Cb \in H_0$, the possible cases are
            \begin{itemize}
                \item[a)] $\Cbd \in H_0$ and $\v_{19,b}(\Cbd)=3$, then $\v_{19,b}(\Cb)=3,\ B^{(3)}=3,\ \v_{19,b} \in C_3$,
                \item[b)] $\Cbd \in H_1$ and $\v_{19,b}(\Cbd)=1$, then $\v_{19,b}(\Cb)=3,\ B^{(1)}=1,\ B^{(3)}=2,\ \v_{19,b} \in C_5$,
                \item[c)] $\Cbd \in H_1$ and $\v_{19,b}(\Cbd)=3$, then $\v_{19,b}(\Cb)=3,\ B^{(1)}=1,\ B^{(3)}=2, \ \v_{19,b} \in C_5$.
            \end{itemize}
            Let $\Cb \in H_1$, the possible cases are
            \begin{itemize}
                \item[a)] $\Cbd \in H_0$ and $\v_{19,b}(\Cbd)=3$, then $\v_{19,b}(\Cb)=3,\ B^{(1)}=2,\ B^{(3)}=1, \ \v_{19,b}\in C_8$,
                \item[b)] $\Cbd \in H_1$ and $\v_{19,b}(\Cbd)=1$, then $\v_{19,b}(\Cb)=1,\ B^{(1)}=2, \ B^{(3)}=1, \ \v_{19,b} \in C_5$,
                \item[c)] $\Cbd \in H_1$ and $\v_{19,b}(\Cbd)=3$, then $\v_{19,b}(\Cb)=3,\ B^{(1)}=1,\ B^{(3)}=2, \ \v_{19,b} \in C_8$.
            \end{itemize}
            Therefore, configuration $\v_{19}$ a weakly periodic ground state on the set $A_3 \cap A_5 \cap A_8=A_5=A_8$.

    \item[20)] Consider
            \begin{equation*}
            \v_{20}(x)=\left\{
            \begin{array}{llll}
            1&if&\xd \in H_0,& x \in H_0,\\
            3&if&\xd \in H_0,& x \in H_1,\\
            3&if&\xd \in H_1,& x \in H_0,\\
            1&if&\xd \in H_1,& x \in H_1.
            \end{array}\right.
            \end{equation*}
            Let $\Cb \in H_0$, the possible cases are
            \begin{itemize}
                \item[a)] $\Cbd \in H_1$ and $\v_{20,b}(\Cbd)=1$, then $\v_{20,b}(\Cb)=3,\ B^{(1)}=3,\ \v_{20,b} \in C_3$,
                \item[b)] $\Cbd \in H_0$ and $\v_{20,b}(\Cbd)=1$, then $\v_{20,b}(\Cb)=1,\ B^{(1)}=2,\ B^{(3)}=1,\ \v_{20,b} \in C_5$,
                \item[c)] $\Cbd \in H_0$ and $\v_{20,b}(\Cbd)=3$, then $\v_{20,b}(\Cb)=1,\ B^{(1)}=2,\ B^{(3)}=1, \ \v_{20,b} \in C_5$.
            \end{itemize}
            Let $\Cb \in H_1$, the possible cases are
            \begin{itemize}
                \item[a)] $\Cbd \in H_0$ and $\v_{20,b}(\Cbd)=1$, then $\v_{20,b}(\Cb)=3,\ B^{(1)}=3,\ B^{(3)}=0, \ \v_{20,b} \in C_3$,
                \item[b)] $\Cbd \in H_1$ and $\v_{20,b}(\Cbd)=1$, then $\v_{20,b}(\Cb)=3,\ B^{(1)}=2, \ B^{(3)}=1, \ \v_{20,b} \in C_8$,
                \item[c)] $\Cbd \in H_1$ and $\v_{20,b}(\Cbd)=3$, then $\v_{20,b}(\Cb)=1,\ B^{(1)}=1,\ B^{(3)}=2, \ \v_{20,b} \in C_5$.
            \end{itemize}
            Therefore, configuration $\v_{20}$ a weakly periodic ground state on the set $A_3 \cap A_5 \cap A_8=A_5=A_8$.

    \end{itemize}

Let we prove assertion III. We consider the remaining configurations. For example, let:
\begin{equation*}
            \v(x)=\left\{
            \begin{array}{llll}
            3&if&\xd \in H_0,& x \in H_0,\\
            1&if&\xd \in H_0,& x \in H_1,\\
            3&if&\xd \in H_1,& x \in H_0,\\
            3&if&\xd \in H_1,& x \in H_1.
            \end{array}\right.
\end{equation*}

We show that $\v(x)$ is not a ground state.

Let $\Cb \in H_0$, the possible cases are
            \begin{itemize}
                \item[a)] $\Cbd \in H_0$ and $\v_{b}(\Cbd)=3$, then $\v_{b}(\Cb)=3,\ B^{(1)}=1,\ B^{(3)}=2,\ \v_b \in C_8$,
                \item[b)] $\Cbd \in H_1$ and $\v_{b}(\Cbd)=3$, then $\v_{b}(\Cb)=3,\ B^{(1)}=0,\ B^{(3)}=3,\ \v_b \in C_1$.
            \end{itemize}
            Let $\Cb \in H_1$, the possible cases are
            \begin{itemize}
                \item[a)] $\Cbd \in H_0$ and $\v_{b}(\Cbd)=3$, then $\v_{b}(\Cb)=1,\ B^{(1)}=0,\ B^{(3)}=3, \ \v_b \in C_3$,
                \item[b)] $\Cbd \in H_1$ and $\v_{b}(\Cbd)=1$, then $\v_{b}(\Cb)=3,\ B^{(1)}=1, \ B^{(3)}=2, \ \v_b \in C_8$,
                \item[c)] $\Cbd \in H_1$ and $\v_{b}(\Cbd)=3$, then $\v_{b}(\Cb)=3,\ B^{(1)}=0,\ B^{(3)}=3, \ \v_b \in C_1$.
            \end{itemize}

We conclude that the configuration $\v(x)$ is a ground state on the set $\{(a,b,c)\in R^3: a=b=c\}.$ Therefore, if $(a, b, c)\notin \{(a,b,c)\in R^3: a=b=c\}$, then the weakly periodic configuration $\v(x)$ is not a weakly periodic ground state.

As above, we analyze all possible cases, except for periodic (in particular translation-invariant) configurations and the configurations mentioned in assertions I and II we find that none of them is a ground state. This finishes the proof of Theorem \ref{th2}.
\end{proof}

\begin{rem}
    Notice that the all configurations on the set $\{(a,b,c)\in R^3: a=b=c\}$ are ground state.
\end{rem}

\textbf{Aknowledgments.} The author RM thanks IIUM for providing financial
support (grant FRGS 14-116-0357) and all facilities. The author also thanks the organizers of the 37-th International Conference on Quantum Probability and Related Topics (QP37)  for travel support.

\smallskip

\end{document}